\documentclass[10pt,journal,cspaper,compsoc]{IEEEtran}
\ifCLASSOPTIONcompsoc
\else
\fi
\usepackage{amsmath}
\ifCLASSINFOpdf
  \usepackage[pdftex]{graphicx}
  \DeclareGraphicsExtensions{.eps,.pdf,.jpeg,.png}
\else
   \usepackage[dvips]{graphicx}
   \DeclareGraphicsExtensions{.eps}
\fi
\newcommand{\commentOut}[1]{}

\newtheorem{theorem}{Theorem}
\newtheorem{lemma}{Lemma}
\newtheorem{problem}{Problem}

\begin{document}
%
\title{On string matching with k mismatches}
%
%
%
%
\author{Marius~Nicolae~and~Sanguthevar~Rajasekaran
\IEEEcompsocitemizethanks{\IEEEcompsocthanksitem M. Nicolae and S. Rajasekaran
are with the Department of Computer Science and Engineering, University
of Connecticut, Storrs, CT, 06269.\protect\\
E-mail: man09004@engr.uconn.edu and rajasek@engr.uconn.edu 
}
\thanks{}}

%
%

\markboth{Nicolae and Rajasekaran: On string matching with k mismatches}{}
%


\IEEEcompsoctitleabstractindextext{%
\begin{abstract}

In this
paper we consider several variants of the pattern matching problem. In
particular, we investigate the following problems: 1) Pattern matching with $k$
mismatches; 2) Approximate counting of mismatches; and 3) Pattern matching with
mismatches. The distance metric used is the Hamming distance. 
We present some novel algorithms and techniques for solving these
problems. Both deterministic and randomized algorithms are offered. Variants of
these problems where there could be wild cards in either the text or the
pattern or both are considered. An experimental evaluation of these algorithms
is also presented. The source code is available at
http://www.engr.uconn.edu/$\sim$man09004/kmis.zip.

\end{abstract}

}

\maketitle

\IEEEdisplaynotcompsoctitleabstractindextext

%
\IEEEpeerreviewmaketitle

\section{Introduction}
%
%

%
%
%
%
\IEEEPARstart{T}{he} problem of string matching has been studied extensively due to its wide
range of applications from Internet searches to computational biology. 
The simplest version takes as input a text
$T=t_1 t_2\cdots t_n$ and a pattern $P=p_1p_2\cdots p_m$ from an alphabet
$\Sigma$. The problem is to find all the occurrences of the pattern in the text.
Algorithms for solving this problem in $O(n+m)$ time are
well known (e.g.,
\cite{KMP77}). A variation of this problem searches for multiple patterns at
the same time (e.g. \cite{AC75}). A more general version allows for ``don't
care'' or ``wild card'' characters (they match any character) in the text and the
pattern. A simple $O(n \log |\Sigma| \log m)$ algorithm for pattern matching
with wild cards is given in \cite{FP74}. 
\commentOut{Each character in $\Sigma$ is mapped to a binary code of length
$\log |\Sigma|$ then convolutions are used to find all the instances of the
pattern in the text.}
A randomized $O(n \log n)$
algorithm which solves the problem with high probability is given in
\cite{IND98}. A slightly faster randomized $O(n \log m)$ algorithm is given in
\cite{Kal02}. A simple deterministic $O(n \log m)$ algorithm based on
convolutions is given in \cite{CC07}.

A more challenging instance of the problem is pattern matching with mismatches.
There are two versions: a) for every alignment of the pattern in the text, find the
distance between the pattern and the text, or b) identify only those alignments
where the distance between the pattern and the text is less than a given
threshold. The distance metric can be the Hamming distance, edit
distance, L1 metric, and so on. In \cite{CH99} the problem has been
generalized to use trees instead of sequences or to use sets of characters instead of single
characters. The Hamming distance between two strings of equal 
length is defined as the  number of positions where the two strings differ. In this paper 
we are interested in the following two problems, with and without wild cards.

{\bf 1. Pattern matching with mismatches:} Given a text $T=t_1t_2\ldots
t_n$, and a pattern $P=p_1p_2\ldots p_m$, output, for every $i,~1\leq i\leq (n-m+1)$, 
the Hamming distance between $t_i t_{i+1}\ldots t_{i+m-1}$ and $p_1p_2\ldots p_m$.

{\bf 2. Pattern matching with $k$ mismatches (or the $k$-mismatches problem):}
Take an additional input parameter $k$. Output all $i,~1\leq i\leq (n-m+1)$
for which the Hamming distance between $t_i t_{i+1},\ldots t_{i+m-1}$
and $p_1p_2\ldots p_m$ is less or equal to $k$.

\subsection{Pattern matching with mismatches}
For pattern matching with mismatches, the naive algorithm
computes the Hamming distance for every alignment of the pattern in the
text, in time $O(nm)$. A faster algorithm is Abrahamson's algorithm, which runs in time $O(n \sqrt{m \log m})$.
We prove that this algorithm can be modified to obtain an $O(n \sqrt{g \log m})$ time algorithm 
for pattern matching with mismatches and wild cards, where $g$ is the
number of non-wild card positions in the pattern. This gives a simpler and
faster alternative to an algorithm proposed in \cite{ALP04}.

In the literature, we also find algorithms which approximate the number of mismatches for every 
alignment. For example, \cite{ACD01} gives an $O(r n \log m)$ time
algorithm for pattern matching with mismatches, in the absence of wild cards, where
$r$ is the number of iterations of the
algorithm. Every distance reported has a variance bounded by $(m-c_i)/{r^2}$ where $c_i$ 
is the exact number of matches for alignment $i$.
\cite{KAR93} gives an $O(n \log^cm/\epsilon^2)$ time randomized algorithm which approximates the
Hamming distance for every alignment within an $\epsilon$ factor, in the absence of wild cards. 
We show how to extend this algorithm to pattern matching with mismatches {\em and} wild cards. 
The new algorithm
approximates the Hamming distance for every alignment within an $\epsilon$ factor in time
$O(n\log^2m/\epsilon^2)$ with high probability.

\subsection{Pattern matching with $k$ mismatches}
For the $k$-mismatches problem, without wild cards, $O(nk)$ time
algorithms are presented in \cite{LV85,GG86}. A faster $O(n\sqrt{k \log k})$ 
time algorithm is presented in
\cite{ALP04}. This algorithm combines the two main techniques known in the
literature for pattern matching with mismatches: filtering and convolutions. 
We give a significantly simpler algorithm having the
same worst case run time. Furthermore, the new algorithm will never perform more
operations than the one in \cite{ALP04} during marking and convolution. 
\commentOut{because
we use information about the frequency of the characters in the text to
minimize the time rather than make it fit within the worst case bound.}

An intermediate problem is to check if the
Hamming distance is less or equal to $k$ for a subset of the aligned positions.
This problem can be solved with the Kangaroo method proposed in \cite{ALP04} at a cost
of $O(k)$ time per alignment, using $O(n+m)$ additional memory. We show how to achieve the same
run time per alignment using only $O(m)$ additional memory.

Further, we look at the version of $k$-mismatches where wild cards are allowed
in the text and the pattern. For this problem, two randomized algorithms are
presented in \cite{CEP+07}. The first one runs in $O(nk\log n\log m)$  time and
the second one in $O\left (n\log m(k+\log n\log\log n)\right )$ time.
Both are Monte Carlo algorithms, i.e. they output the correct
answer with high probability. The same paper also gives a deterministic
algorithm with a run time of $O(nk^2\log^3m)$. Also, a deterministic $O(nk \log^2
m(\log^2 k + \log \log m))$ time algorithm is given in \cite{CEPR09}. We present a
Las Vegas algorithm (that always outputs the correct answer) which runs in time
$O(nk\log^2 m+n\log^2m\log n+n\log m\log n\log\log n)$ with high probability.

The contributions of this paper can be summarized as follows.

For pattern matching with mismatches:
\begin{itemize}
\item An $O(n \sqrt{g \log m})$ time algorithm  for pattern matching with mismatches
and wild cards, where $g$ is the number of non-wild card positions in the
pattern.

\item A randomized algorithm that approximates the Hamming distance for every
alignment, when wild cards are present, within an $\epsilon$ factor in time
$O(n\log^2m/\epsilon^2)$ with high probability.
\end{itemize}

For pattern matching with $k$ mismatches:
\begin{itemize}
\item An $O(n\sqrt{k \log k})$ time algorithm for pattern matching
with $k$ mismatches, without wild cards, which is simpler and has a better expected run
time than the one in \cite{ALP04}.

\item An algorithm that tests if the Hamming distance is less than $k$
for a subset of the alignments, without wild cards, at a cost of $O(k)$ time per
alignment, using only $O(m)$ additional memory.

\item A Las Vegas algorithm for the
$k$-mismatches problem with wild cards that runs in
time $O(nk\log^2 m+n\log^2m\log n+n\log m\log n\log\log n)$ with high probability.
\end{itemize}

The rest of the paper is organized as follows. First we introduce some notations
and definitions. Then we present the randomized and approximate algorithms:
first the Las Vegas algorithm for $k$-mismatches with wild cards, then the
algorithm for approximate counting of mismatches in the presence of wild
cards.
Then we describe the deterministic and exact algorithms, for which we also present an
empirical run time comparison.

 \section{Some definitions}
Given two strings $T=t_1t_2\ldots t_n$ and $P=p_1p_2\ldots
p_m$ (with $m\leq n$), the convolution of $T$ and $P$ is a sequence
$C=c_1,c_2,\ldots,c_{n-m+1}$ where $c_i=\sum_{j=1}^mt_{i+j-1}p_j$, for $1\leq
i\leq (n-m+1)$. We can compute this convolution in $O(n\log m)$ time using the
Fast Fourier Transform. Some speedup techniques exist \cite{FG09}
for convolutions applied on binary inputs, as is usually the case 
with pattern matching.

In the context of randomized algorithms, by
high probability we mean a probability greater or equal to $(1-n^{-\alpha})$ 
where $n$ is the input size and $\alpha$ is a probability parameter usually
assumed to be a constant greater than $0$. The run time of a Las Vegas algorithm
is said to be $\widetilde O(f(n))$ if the run time is no more than $c\alpha f(n)$
with probability greater or equal to $(1-n^{-\alpha})$ for all $n\geq n_0$,
where $c$ and $n_0$ are some constants, and for any constant $\alpha\geq 1$.

\section{A Las Vegas algorithm for $k$-mismatches}

\subsection{The 1-Mismatch Problem}
\noindent{\bf Problem Definition:} For this problem also, the input are two
strings $T$ and $P$ with $|T|=n,|P|=m,$ and $m\leq n$. Let $T_i$ stand for the
substring $t_{i}t_{i+1}\ldots t_{i+m-1}$, for any $i$, with $1\leq i\leq
(n-m+1)$. The problem is to check if the Hamming distance between $T_i$ and $P$
is exactly 1, for $1\leq i\leq (n-m+1)$. The following Lemma is shown in
\cite{CEP+07}.

\begin{lemma}\label{1mm}
The 1-mismatch problem can be solved in $O(n\log m)$ time using a constant
number of convolution operations.
\end{lemma}

 \noindent{\bf The Algorithm:} Assume that each wild card in the pattern as
 well as the text is replaced with a zero. Also, assume that the characters in
 the text as well as the pattern are integers in the range $[1:|\Sigma|]$ where
 $\Sigma$ is the alphabet under concern. Let $e_{i,j}$ stand for the ``error
 term" introduced by the character $t_{i+j-1}$ in $T_i$ and the character $p_j$
 in $P$ and its value is $(t_{i+j-1}-p_j)^2t_{i+j-1}p_j$. Also, let
 $E_i=\sum_{j=1}^me_{i,j}$. There are four steps in the algorithm:
 \begin{enumerate}
 \item Compute $E_i$ for $1\leq i\leq(n-m+1)$. Note that $E_i$ will be zero if
 $T_i$ and $P$ match (assuming that a wild card can be matched with any character).
 \item Compute $E'_i$ for $1\leq i\leq (n-m+1)$, where
 $E'_i=\sum_{j=1}^m(i+j-1)(t_{i+j-1}-p_j)^2p_jt_{i+j-1}$ (for $1\leq i\leq(n-m+1))$.
 \item Let $B_i=E'_i/E_i$ if $E_i\neq 0$, for $1\leq i\leq (n-m+1)$. Note that
 if the Hamming distance between $T_i$ and $P$ is exactly one, then $B_i$ will
 give the position in the text where this mismatch occurs.
 \item If for any $i$ ($1\leq i\leq (n-m+1)$), $E_i\neq 0$ and if
 $(t_{B_i}-p_{B_i-i+1})^2t_{B_i}p_{B_i-i+1}=E_i$ then we conclude that the
 Hamming distance between $T_i$ and $P$ is exactly one.
 \end{enumerate}
 
 \noindent{\bf Note:} If the Hamming distance between $T_i$ and $P$ is exactly
 1 (for any $i$), then the above algorithm will not only detect it but also
 identify the position where there is a mismatch. Specifically, it will
 identify the integer $j$ such that $t_{i+j-1}\neq p_j$.

\subsection{The Randomized Algorithms of \cite{CEP+07}}
Two different randomized algorithms are presented in \cite{CEP+07} for solving
the $k$-mismatches problem. Both are Monte Carlo algorithms. In particular, they
output the correct answers with high probability. The run times of these
algorithms are $O(nk\log m\log n)$ and $O(n\log m(k+\log n\log\log n))$,
respectively. In this section we provide a summary of these algorithms.

The first algorithm has $O(k\log n)$ sampling phases and in each phase a
1-mismatch problem is solved. Each phase of sampling works as follows. We
choose $m/k$ positions of the pattern uniformly at random. The pattern $P$ is
replaced by a string $P'$ where $|P'|=m$, the characters in $P'$ in the
randomly chosen positions are the same as those in the corresponding positions
of $P$, and the rest of the characters in $P'$ are chosen to be wild cards. The
1-mismatch algorithm of Lemma \ref{1mm} is run on $T$ and $P'$. In each phase
of random sampling, for each $i$, we get to know if the Hamming distance
between $T_i$ and $P'$ is exactly 1 and, if so, identify the $j$ such that
$t_{i+j-1}\neq p_j^\prime$.

As an example, consider the case when the Hamming distance between $T_i$ and
$P$ is $k$ (for some $i$). Then, in each phase of sampling we would expect to
identify exactly one of the positions (i.e., $j$) where $T_i$ and $P$ differ
(i.e., $t_{i+j-1}\neq p_j$). As a result, in an expected $k$ phases of sampling
we will be able to identify all the $k$ positions in which $T_i$ and $P$
differ. It can be shown that if we make $O(k\log n)$ sampling phases, then we
can identify all the $k$ mismatches with high probability \cite{CEP+07}.

Let the number of mismatches between $T_i$ and $P$ be $q_i$ (for $1\leq i\leq
(n-m+1)$. If $q_i\leq k$, the algorithm of \cite{CEP+07} will compute $q_i$
exactly. If $q_i>k$, then the algorithm will report that the number of
mismatches is $>k$ (without estimating $q_i$) and this answer will be correct
with high probability. The algorithm starts off by first computing $E_i$ values
for every $T_i$. A list $L(i)$ of all the mismatches found for $T_i$ is kept,
for every $i$. Whenever a mismatch is found between $T_i$ and $P$ (say in
position $(i+j-1)$ of the text), the value of $E_i$ is reduced by $e_{i,j}$. If
at any point in the algorithm $E_i$ becomes zero for any $i$ it means that we
have found all the $q_i$ mismatches between $T_i$ and $P$ and $L(i)$ will have
the positions in the text where these mismatches occur. Note that if the
Hamming distance between $T_i$ and $P$ is much larger than $k$ (for example
close or equal to $m$), then the probability that in a random sample we isolate
a single mismatch is very low. Therefore, if the number of sample phases is
only $O(k\log n)$, the algorithm can only be Monte Carlo. Even if $q_i$ is
$\leq k$, there is a small probability that we may not be able to find all the
$q_i$ mismatches. Call this algorithm {\bf Algorithm 1}. If for each $i$, we
either get all the $q_i$ mismatches (and hence the corresponding $E_i$ is zero)
or we have found $>k$ mismatches between $T_i$ and $P$ then we can be sure that
we have found all the correct answers (and the algorithm will become Las
Vegas).

The authors of \cite{CEP+07} also present an improved algorithm whose run time
is $O(n\log m(k+\log n\log\log n))$. The main idea is the observation that if
$q_i= k$ for any $i$, then in $O(k\log n)$ sampling steps we can identify $\geq
k/2$ mismatches. There are several iterations where in each iteration $O(k+\log
n)$ sampling phases are done. At the end of each iteration the value of $k$ is
changed to $k/2$. Let this algorithm be called {\bf Algorithm 2}.
  
\subsection{A Las Vegas Algorithm}
In this section we present a Las Vegas algorithm for the $k$-mismatches problem
when there are wild cards in the text and/or the pattern. This algorithm runs in
time $\widetilde O(nk\log^2m+n\log^2m\log n+n\log m\log n\log\log n)$. This
algorithm is based on the algorithm of \cite{CEP+07}. When the algorithm
terminates, for each $i$ ($1\leq i\leq (n-m+1))$, either we would have
identified all the $q_i$ mismatches between $T_i$ and $P$ or we would have
identified more than $k$ mismatches between $T_i$ and $P$.

{\bf Algorithm 1} will be used for every $i$ for which $q_i\leq 2k$. For every
$i$ for which $q_i>2k$ we use the following strategy. Let $2^\ell k<q_i\leq
2^{\ell+1}k$ (where $1\leq \ell\leq \log\left
(\left\lfloor\frac{m}{2k}\right\rfloor\right )$). Let $w=\log\left
(\left\lfloor\frac{m}{2k}\right\rfloor\right )$. There will be $w$ phases in
the algorithm and in each phase we perform $O(k)$ sampling steps. Each sampling
step in phase $\ell$ involves choosing $\frac{m}{2^{\ell+1} k}$ positions of
the pattern uniformly at random (for $1\leq \ell\leq w$). As we show below, if
for any $i$, $q_i$ is in the interval $[2^\ell,2^{\ell+1}]$, then at least $k$
mismatches between $T_i$ and $P$ will be found in phase $\ell$ with high
probability. A pseudocode for the algorithm (call it {\bf Algorithm 3}) is
given below. An analysis will follow.

\begin{tabbing}
aa \= aa \= aa \= aa \= aa \= \kill
{\bf Algorithm 3}\\
{\bf repeat}\\
{\bf 1.} \> Run {\bf Algorithm 1} or {\bf Algorithm 2}\\ 
{\bf 2.} \> {\bf for} $\ell:=1$ {\bf to} $w$ {\bf do}\\
\> {\bf for} $r:=1$ {\bf to} $ck$ ($c$ being a constant) {\bf do}\\
\>\> Uniformly randomly choose $\frac{m}{2^{\ell+1} k}$ positions\\
\>\> of the pattern;\\
\>\> Generate a string $P'$ such that $|P'|=|P|$\\
\>\> and $P'$ has the same characters as $P$ in\\
\>\> these randomly chosen positions and\\
\>\> zero everywhere else;\\
\>\> Run the 1-mismatch algorithm on $T$ and $P'$.\\
\>\> As a result, if there is a single mismatch\\ 
\>\>  between $T_i$ and $P'$ then add the position of\\
\>\>  mismatch to $L(i)$ and reduce the value of $E_i$ \\
\>\>  by the right amount, for $1\leq i\leq (n-m+1)$;\\
{\bf 3.} \> {\bf if} either $E_i=0$ or $|L(i)|>k$ for every $i$,\\
\>\> $1\leq i\leq (n-m+1)$ {\bf then} quit;\\
{\bf forever}
\end{tabbing}

\begin{theorem}
{\bf Algorithm 3} runs in time $\widetilde O(nk\log^2m+n\log^2m\log n$
$+ n\log m\log n\log\log n)$ if {\bf Algorithm 2} is used in step 1. It runs in
time $\widetilde O(nk\log m\log n+nk\log^2m+n\log^2m\log n)$ if step 1 uses {\bf
Algorithm 1}.

\end{theorem}
\begin{proof}
As shown in \cite{CEP+07}, the run time of {\bf Algorithm 1} is $O(nk\log m\log
n)$ and that of {\bf Algorithm 2} is $O(n\log m(k+\log n\log\log n))$. The
analysis will be done with respect to an arbitrary $T_i$. In particular, we
will show that after the specified amount of time, with high probability, we
will either know $q_i$ or realize that $q_i>k$. It will then follow that the
same statement holds for every $T_i$ (for $1\leq i\leq (n-m+1)$.

Consider phase $\ell$ of step 2 (for an arbitrary $1\leq \ell\leq w$). Let
$2^\ell k<q_i\leq 2^{\ell+1}k$ for some $i$. Using the fact that
${{a}\choose{b}}\approx\left (\frac{ae}{b}\right )^b$, the probability of isolating one of the
mismatches in one run of the sampling step is:

\begin{align*}
\frac{{{m-q_i}\choose{m/(2^{\ell+1} k)-1}}q_i}{{{m}\choose{m/(2^{\ell+1}
k)}}}  \geq\frac{{{m-2^{\ell+1}k}\choose{m/(2^{\ell+1} k)-1}}2^\ell
k}{{{m}\choose{m/(2^{\ell+1} k)}}}\geq\frac{1}{2e}
\end{align*}

 As a result, using
Chernoff bounds, it follows that if $13ke$ sampling steps are made in phase
$\ell$, then at least $6k$ of these steps will result in the isolation of
single mismatches (not all of them need be distinct) with high probability
(assuming that $k=\Omega(\log n)$). Moreover, we can see that at least $1.1k$
of these mismatches will be distinct. This is because the probability that
$\leq 1.1k$ of these are distinct is $\leq {{{q_i}\choose{1.1k}}}/{\left
(\frac{1.1k}{q_i}\right )^{6k}}$ $\leq 2^{-2.64k}$ using the fact that $q_i\geq
2k$. This probability will be very low when $k=\Omega(\log n)$.

In the above analysis we have assumed that $k=\Omega(\log n)$. If this is not
the case, in any phase of step 2, we can do $c\alpha\log n$ sampling steps, for
some suitable constant $c$. As a result, each phase of step 2
takes $O(n\log m(k+\log n))$ time. We have $O(\log m)$ phases. Thus the run time of
step 2 is $O(n\log^2m(k+\log n))$. Also, the probability that the condition in
step 3 holds is very high.

Therefore, the run time of the entire algorithm is 
$\widetilde O(nk\log^2m+n\log^2m\log n$ $+ n \log m\log n\log\log n)$ if {\bf
Algorithm 2} is used in step 1 or $\widetilde O(nk\log m\log
n+nk\log^2m+n\log^2m\log n)$ if {\bf Algorithm 1} is used in step 1.
\end{proof}

\section{Approximate Counting of Mismatches}

The algorithm of \cite{KAR93} takes as input a text $T=t_1t_2\ldots t_n$ and a
pattern $P=p_1p_2\ldots p_m$ and approximately counts the Hamming distance
between $T_i$ and $P$ for every $1\leq i\leq (n-m+1)$. In particular, if the
Hamming distance between $T_i$ and $P$ is $H_i$ for some $i$, then the
algorithm outputs $h_i$ where $H_i\leq h_i\leq (1+\epsilon)H_i$ for any
$\epsilon>0$ with high probability (i.e., a probability of
$\geq(1-m^{-\alpha}))$. The run time of the algorithm is
$O(n\log^2m/\epsilon^2)$. In this section we show how to extend this algorithm
to the case where there could be wild cards in the text and/or the pattern.

Let $\Sigma$ be the alphabet under concern and let $\sigma=|\Sigma|$. 
The algorithm runs in phases and in each phase we randomly map the elements of
$\Sigma$ to $\{1,2\}$. A wild card is mapped to a zero. Under this mapping we
transform $T$ and $P$ to $T'$ and $P'$, respectively.  We then compute a vector
$C$ where $C[i]=\sum_{j=1}^m (t'_{i+j-1}-p'_j)^2t'_{i+j-1}p'_j$. This can be
done using $O(1)$ convolution operations (see e.g., \cite{CEP+07}). A series of
$r$ such phases (for some relevant value of $r$) is done at the end of which we
produce estimates on the Hamming distances. The intuition is that if a
character $x$ in $T'$ is aligned with a character $y$ in $P'$, then across all
the $r$ phases, the expected contribution to $C$ from these characters is $r$
if $x\neq y$ (assuming that $x$ and $y$ are non wild cards).  If $x=y$ or if
one or both of $x$ and $y$ are a wild card, the contribution to $C$ is zero.

\begin{tabbing}
aa \= aa \= aa \= aa \= aa \= \kill
\>{\bf Algorithm 4}\\
\>{\bf 1.} \> {\bf for} $i:=1$ {\bf to} $(n-m+1)$ {\bf do} $C[i]=0$.\\
\>{\bf 2.} \> {\bf for} $\ell:=1$ {\bf to} $r$ {\bf do}\\
\>\>\> Let $Q$ be a random mapping of $\Sigma$ to $\{1,2\}$.\\
\>\>\> In particular, each element of $\Sigma$ is mapped \\
\>\>\> to 1 or 2 randomly with equal probability.\\
\>\>\>   Each wild card is mapped to a zero.\\
\>\>\> Obtain two strings $T'$ and $P'$
where $t_i'=Q(t_i)$\\
\>\>\> for $1\leq i\leq n$ and $p_j'=Q(p_j)$ for $1\leq j\leq m$.\\
\>\>\> Compute a vector $C_\ell$ where\\
\>\>\>$C_\ell[i]=\sum_{j=1}^m(t'_{i+j-1}-p'_j)^2~ t_{i+j-1}'p_j'$ \\
\>\>\>\> for $1\leq i\leq (n-m+1)$.\\
\>\>\> {\bf for} $i:=1$ {\bf to} $(n-m+1)$ {\bf do} $C[i]:=C[i]+C_\ell[i]$.\\
\>{\bf 3.} \> {\bf for} $i:=1$ {\bf to} $(n-m+1)$ {\bf do}\\
\>\>\> output
$h_i:=\frac{C[i]}{r}$.\\
\>Here $h_i$ is an estimate on the Hamming distance\\
\> $H_i$ between $T_i$ and $P$.
\end{tabbing}

\noindent{\bf Analysis:} Let $x$ be a character in $T$  and let $y$ be a
character in $P$. Clearly, if $x=y$ or if one or both of  $x$ and $y$ are a wild
card, the contribution of $x$ and $y$ to any $C_\ell[i]$  is zero. If $x$ and
$y$ are non wild cards and if $x\neq y$ then the expected  contribution of these
to any $C_\ell[i]$ is 1. Across all the $r$ phases, the  expected contribution
of $x$ and $y$ to any $C_\ell[i]$ is $r$. For a given $x$  and $y$, we can think
of each phase as a Bernoulli trial with equal probabilities  for success and
failure. A success refers to the possibility of $Q(x)\neq Q(y)$.  The expected
number of successes in $r$ phases is $\frac{r}{2}$. Using Chernoff  bounds, this
contribution is no more than $(1+\epsilon)r$ with probability  $\geq
1-\exp(-\epsilon^2r/6)$. Probability that this statement holds  for every pair
$(x,y)$ is $\geq 1-m^2\exp(-\epsilon^2r/6)$. This probability will  be $\geq
1-m^{-\alpha}/2$ if $r\geq \frac{6(\alpha+3)\log_em}{\epsilon^2}$.  Similarly,
we can show that for any pair of non wild card characters,  the contribution of
them to any $C_\ell[i]$ is no less than $(1-\epsilon)r$ with  probability $\geq
1-m^{-\alpha}/2$ if $r\geq\frac{4(\alpha+3)\log_em}{\epsilon^2}$.

Put together, for any pair $(x,y)$ of non wild cards, the contribution of $x$
and $y$ to any $C_\ell[i]$ is in the interval $(1\pm\epsilon)r$ with
probability $\geq(1-m^{-\alpha})$ if
$r\geq\frac{6(\alpha+3)\log_em}{\epsilon^2}$. Let $H_i$ be the Hamming distance
between $T_i$ and $P$ for some $i$ ($1\leq i\leq (n-m+1))$. Then, the estimate
$h_i$ on $H_i$ will be in the interval $(1\pm \epsilon)H_i$ with probability
$\geq(1-m^{-\alpha})$. As a result, we get the following Theorem.

\begin{theorem}
Given a text $T$ and a pattern $P$, we can estimate the Hamming distance
between $T_i$ and $P$, for every $i,~1\leq i\leq (n-m+1)$, in 
$O(n\log^2m/\epsilon^2)$ time. If $H_i$ is the Hamming distance between $T_i$
and $P$, then the above algorithm outputs an estimate that is in the interval
$(1\pm\epsilon) H_i$ with high probability.
\end{theorem}

\noindent{\bf Observation 1.} In the above algorithm we can ensure that $h_i\geq H_i$ and $h_i\leq(1+\epsilon)H_i$ with high probability by changing the estimate computed in step 3 of Algorithm 4 to $\frac{C[i]}{(1-\epsilon)r}$.

\noindent{\bf Observation 2.} As in \cite{KAR93}, with $O\left (\frac{m^2\log m}{\epsilon^2}\right )$ pre-processing we can ensure that Algorithm 4 never errs (i.e., the error bounds on the estimates will always hold). 
\section{Deterministic algorithms}
In this section we present deterministic algorithms for the problems of interest. 
We first summarize two well known techniques for counting
matches: convolution and marking (see e.g. \cite{ALP04}). 
In terms of notation, $T[i]$ is the character at position $i$ in $T$, $T_{i..j}$
is the substring of $T$ between $i$ and $j$ and $T_i$ is $T_{i..i+m-1}$ as
before.

{\bf Convolution:} Given a string $S$ and a character $\alpha$ define $S^{\alpha}$
to be a string where $S^\alpha[i]=1$ if $S[i]=\alpha$ and $0$ otherwise.
Let $C^\alpha=convolution(T^\alpha, P^\alpha)$. Then $C^\alpha[i]$ gives the
number of positions $j$ where $P[j] = T[i+j-1] = \alpha$,
which is the number of matches ``contributed'' by character $\alpha$ to
the alignment between $P$ and $T_i$. Then $\sum_{\alpha \in
\Sigma}C^{\alpha}[i]$ is the number of matches between $P$ and $T_i$.

{\bf Marking:}\label{sec_marking} Given a character
$\alpha$ let $Pos[\alpha] = \{i \in [1..m] | P[i] = \alpha\}$.
Let $\Gamma$ be a subset of $\Sigma$. The number of matches between $P$ and $T_i$ where
the matching character is from $\Gamma$ can be computed by the following algorithm. The number
of matches are reported in $M$. The algorithm takes $O(n \max_{\alpha \in \Gamma}|Pos_\alpha|)$ time.

\begin{center}
\label{alg_counting}
\begin{tabbing}
aa \= aa \= aa \= aa \= aa \= \kill
\>{\bf Algorithm 5}  $Mark(T, n, \Gamma)$\\
\>\> {\bf for} $i:=1$ {\bf to} $n$ {\bf do} $M[i] := 0$\\
\>\> {\bf for} $i:=1$ {\bf to} $n$, {\bf if} $T[i] \in \Gamma$ {\bf do} \\
\>\>\> {\bf for} {$j \in Pos[T[i]]$}, {\bf if} $i-j+1 > 0$ {\bf do}\\
\>\>\>\> $M[i-j+1]${\bf ++} \\
\>\> {\bf return} M
\end{tabbing}
\end{center}

\subsection{Pattern matching with mismatches}
\label{sec_glogm}
For pattern matching with mismatches, without wild cards,
Abrahamson \cite{ABR87} gave the following $O(n\sqrt{m \log m})$ time algorithm.
Let $A$ be a set of the most frequent characters
in the pattern. 1) Using convolutions, count how many matches each character in $A$ contributes to
every alignment. 2) Using marking, count how many matches each character
in $\Sigma - A$, contributes to every alignment. 3) Add the two
numbers to find for every alignment, the number of matches between the pattern
and the text. The convolutions take $O(|A| n \log m)$ time. A
character in $\Sigma - A$ cannot appear more than $m/|A|$ times in the pattern,
otherwise, each character in $A$ has a frequency greater than
$m/|A|$, which is not possible. Thus, the run time for marking is $O(n m / |A|)$.
If we equate the two run times we find the optimal $|A| = \sqrt{m / \log m}$ which gives a
total run time of $O(n \sqrt{m \log m})$.

For pattern matching with mismatches and wild cards, 
a fairly complex algorithm is given in \cite{ALP04}. The run time is $O(n
\sqrt{g} \log m)$ where $g$ is the number of non-wild card positions in the
pattern. The problem can also be solved through a simple modification of 
Abrahamson's algorithm, in time $O(n\sqrt{m \log m})$, as pointed out in 
\cite{CEP+07}. We now prove the following result:

\begin{theorem} Pattern matching with mismatches and wild cards can be solved
in $O(n\sqrt{g \log m})$ time, where $g$ is the number of non-wild card
positions in the pattern.
\end{theorem}

\begin{proof}
 Ignoring the wild cards for now, let $A$ be the set of
the most frequent characters in the pattern. As above, count matches contributed
by characters in $A$ and $\Sigma-A$ using convolution and marking respectively.
By a similar reasoning as above, the characters used in the marking phase will not 
appear more than $g / |A|$ times in the pattern. If we equate the run times for the two 
phases we obtain $O(n \sqrt {g \log m})$ time. We are now left to count how many matches are
contributed by the wild cards. For a string $S$ and a character $\alpha$, define $S^{\neg \alpha}$ as
$S^{\neg \alpha}[i] = 1-S^\alpha[i]$. Let
$w$ be the wild card character. Compute $C = convolution(T^{\neg w}, P^{\neg w})$. Then,
for every alignment $i$, the number of positions that have a wild card either in the
text or the pattern or both, is $m-C[i]$. Add $m-C[i]$ to the previously computed counts and output. 
The total run time is $O(n \sqrt{g \log m})$.
\end{proof}

\subsection{Pattern matching with $k$ mismatches}
For the $k$-mismatches problem, without wild cards, an $O(k(m\log m+n))$ time
algorithm that requires $O(k(m+n))$ additional space is presented in \cite{LV85}. Another algorithm, that takes
  $O(m\log m + kn)$ time and uses only $O(m)$ additional space is presented
in \cite{GG86}. \commentOut{In the latter, a suffix tree of the pattern is built and enhanced to
support lowest common ancestor queries in $O(1)$ time.} We define the following
problem which is of interest in the discussion.

\begin{problem}
{\bf Subset $k$-mismatches:} Given a text $T$ of length $n$, a pattern $P$ of
length $m$, a set of positions $S=\{i|1 \leq i \leq n-m+1\}$ and an integer $k$, output the
positions $i \in S$ for which the Hamming distance between $P$ and $T_i$ is
less or equal to $k$.
\end{problem}

The problem becomes the regular $k$-mismatches
problem if $|S|=n$. However, if $S$ contains only a fraction of all positions,
the $O(nk)$ algorithms mentioned above are too costly.
A better alternative is proposed in \cite{ALP04}: build a suffix tree of
$T\#P$ and enhance it to support LCA queries in $O(1)$ time. Given position
$i$, perform an LCA query to find the position of the first mismatch between $P$
and $T_i$, call it $j$. Then, perform another LCA to find the first mismatch
between $P_{j+1..m}$ and $T_{i+j+1.. i+m-1}$, which is the second mismatch of alignment $i$.
Repeatedly jump from one mismatch to the next, until the end
of the pattern is reached or we have found more than $k$ mismatches. This is
called the Kangaroo method. It can process $|S|$ positions in 
$O(n+m+|S|k)$ time and it uses $O(n+m)$ additional memory for the LCA enhanced suffix tree. 
We prove the following result:

\begin{theorem}
{\bf Subset $k$-mismatches} can be solved in $O(n+m+|S|k)$ time using only
$O(m)$ additional memory.
\end{theorem}

\begin{proof}
  The algorithm  is the following - also see algorithm
$6$ in the appendix. Build an LCA-enhanced suffix tree of the
pattern. 1) Find the longest unscanned region of
the text which can be found somewhere in the pattern. 2) For every alignment
that overlaps this region of the text, count how many mismatches are found in
the overlapping region. To do this, we compare the pattern against itself, by using
LCA queries, because we know that the text is the same as the pattern, in that region. Repeat from step 1 until 
the entire text has been scanned. Every time we process an alignment in step 2, 
we either discover at least one additional mismatch or we reach the end of the alignment.
This is true since in step 1 we always pick the longest portion of text that can be found somewhere in
the pattern.
In addition, every alignment for which we have found more than $k$ mismatches is
excluded from further consideration. This ensures we spend $O(k)$ time per alignment. It
takes $O(m)$ time to build the LCA enhanced suffix tree of the pattern and
$O(n)$ additional time to scan the text from left to right. Thus, the total run time is
$O(n+m+|S|k)$ with $O(m)$ additional memory.
\end{proof}

\subsection{An $O(n\sqrt{k \log k})$ time algorithm for $k$-mismatches}
\label{sec_nsqrtk}
 
For the $k$-mismatches problem, without wild cards, a fairly complex 
$O(n \sqrt{k \log k})$ time algorithm  is given in \cite{ALP04}. 
The algorithm classifies the inputs into several cases. For each case
it applies a combination of marking followed by a filtering step, the Kangaroo
 method, or convolutions. The goal is to not exceed $O(n \sqrt{k \log k})$ time 
 in any of the cases.
We now present an algorithm with only two cases which has the same worst case run time.
The new algorithm can be thought of as a generalization of the algorithm in \cite{ALP04} 
as we will discuss later.
This generalization not only greatly simplifies the algorithm but it also
reduces the expected run time. This happens because we use information about the frequency of the
characters in the text and try to minimize the work done by convolutions and marking.

For any character $\alpha \in \Sigma$, let $f_\alpha$ be its frequency in the
pattern, and $F_\alpha$ be its frequency in the text. Clearly, $\sum_{\alpha
\in \Sigma} f_\alpha = m$ and $\sum_{\alpha \in \Sigma} F_\alpha = n$. A
position $j$ in the pattern where $p[j] = \alpha$ is called an {\em instance} of
$\alpha$.  Consider every instance of character $\alpha$ as an object of size
$1$ and cost $F_\alpha$. We want to fill a knapsack of size $2k$ at a cost less
than a given budget $B$.
This problem can be optimally solved by a greedy approach where we include all the instances
of the least expensive character, then all the instances of the second least expensive
character and so on, until we have $2k$ items or we have exceeded $B$. The last
character considered may have only a subset of its instances included, but for
ease of explanation assume that there are no such characters.

The algorithm is the following: Case 1) If we can fill the knapsack within budget
$B$, we apply the marking algorithm for the characters whose instances are
included in the knapsack. If alignment $i$ matches perfectly, we
will obtain exactly $2k$ marks at position $i$ in the text. Thus, any position
which has less than $k$ marks must  have more than
$k$ mismatches. Based on this observation, we run Subset $k$-mismatches to
check only those positions with at least $k$ marks.

Case 2) If we cannot fill the knapsack within the given budget we do the following:
for the characters we could fit in the knapsack before we ran out of
budget, we use the marking algorithm to count the number of matches they 
contribute to each alignment. For characters not in the knapsack, we use
convolutions to count the number of matches they contribute to each alignment.
We add the two counts and get the exact number of matches for every alignment. We call
this algorithm {Knapsack $k$-mismatches} (also see algorithm 7 in the
appendix).

\begin{theorem}
{\bf Knapsack $k$-mismatches} has worst case run time $O(n \sqrt{k \log
k})$.
\end{theorem}

\begin{proof}
In case 1, if we can fill the
knapsack within budget $B$, we apply the marking algorithm. This takes
$\sum_{\alpha \in knapsack}f_\alpha F_\alpha = B$ time and creates just as many
marks. Thus, there will be no more than $B/k$ positions with at least $k$ marks. 
We run Subset $k$-mismatches for these positions and obtain a run time of
$O(n+m+B)$.

In case 2, if we cannot fill the knapsack within the given budget, we apply the
marking algorithm for whatever items we could fit in the knapsack. This takes
 $O(B)$ time. Note that
if we add the costs of including  in the knapsack all the instances of
characters with frequency lower than $B/n$ we get $\sum_{f_\alpha < B/n}f_\alpha
F_\alpha < B/n\sum_\alpha F_\alpha = B$. We can include all of
them in the knapsack by only adding a constant factor to the run time of the marking stage.
Thus, we can assume that the characters not in the knapsack have frequency $f_\alpha \geq B/n$.
There cannot be more than $r=2k/(B/n)$ characters not in the knapsack, otherwise
we could have filled the knapsack within budget $B$ by picking $B/n$ instances
for each of $r$ such characters, for a total of $2k$ positions and a cost
$\sum_{i=1}^{r} B/n F_i \leq B$. Thus, it takes $O(r n \log m)=O(n^2k\log m/B)$ time to compute
convolutions for the characters not in the knapsack. If we make this cost equal to the
cost of the marking phase, $O(B)$, we find $B = n \sqrt{k \log m}$.
As in \cite{ALP04}, if $k < m^{1/3}$  we can employ a different algorithm
which solves the problem in linear time. For larger $k$, $O(\log m) = O(\log k)$
so the run time becomes $O(n \sqrt{k \log k})$.
\end{proof}

We can think of the algorithm in \cite{ALP04} as a special case of our
algorithm where, instead of trying to minimize the cost of the $2k$ items in the
knapsack, we just try to find $2k$ of them for which the cost is less than
$O(n\sqrt{k \log m})$. As a result, it is easy to verify the following:
\begin{theorem}\label{knapsack_less}
{\bf Knapsack $k$-mismatches} spends at most as much time as
the algorithm in \cite{ALP04} to do convolutions and marking.
\end{theorem}
\begin{proof}
In the appendix.
\end{proof}

\commentOut{
\subsection{Faster convolutions}
\label{sec_speedup}
The convolutions used so far take as input binary vectors. A speedup
technique for such cases is given in \cite{FG09}. The time for a convolution
is reduced to $O(n \log^2m / w)$ where $w$ is the size of the computer word, in
bits.
It is useful to look at the run time of various algorithms if faster convolution
algorithms are used. Let $T_C$ be the run time to perform a convolution.
Abrahamson's algorithm runs in time $O(\sqrt{n m T_C})$, the pattern
matching with wild cards algorithm described in section \ref{sec_glogm} runs in
time $O(\sqrt{n g T_C})$ and {\bf Knapsack $k$-mismatches} has a run time of
$O(\sqrt{n k T_C})$.
With the above speedup technique these run times become $O(n \log m \sqrt{m/w})$ for
Abrahamson's algorithm, $O(n \log m \sqrt{g / w})$ for the algorithm in section
\ref{sec_glogm} and $O(n \log m \sqrt{k/w})$ for  {\bf Knapsack $k$-mismatches}
(and since for small $k$ we use a different algorithm, the run time is $O(n \log k \sqrt{k/w})$).
}

\section{Experimental Results}
It is interesting to analyze how some of the above algorithms compare in
practice, since some of them are based on symbol comparison, some on arithmetic
operations, and some on a combination of both. We implemented the following
algorithms:
the naive $O(nm)$ time algorithm, Abrahamson's, Subset $k$-mismatches and
Knapsack $k$-mismatches.
For Subset $k$-mismatches, we simulate
the suffix tree and LCA extensions by a suffix array with an LCP (Longest Common Prefix \cite{KLA+01})
table and data structures to perform RMQ queries (Range
Minimum Queries \cite{BFC00}) on it. This adds a $O(\log n)$ factor to
preprocessing and searching. However, faster implementations are possible. 
For Subset $k$-mismatches, we also tried a simple $O(m^2)$ time
pre-processing using dynamic programming and hashing. Knapsack $k$-mismatches
uses Subset $k$-mismatches as a subroutine, so we have two versions of it also. 
We use all algorithms to solve the
$k$-mismatches problem, even though some are more general.

We tested the algorithms on protein, DNA and English inputs from the Pizza
$\&$ Chili Corpus \cite{PizzaChili}. These inputs were truncated at several thresholds, to analyze how 
run time varies with the length of the text. We randomly selected a substring of
length $m$ from the text and used it as pattern. The algorithms were tested on
an Intel Core i3 machine with 4GB of RAM, Ubuntu 11.10 Operating System and gcc
4.6.1. All convolutions were performed using the fftw \cite{FFTW05} library.

 Figure \ref{fig_runtimes} shows run times for varying $n, m, k$ and $
 |\Sigma|$. The naive algorithm performed well in practice most likely due to
 its simplicity and cache locality.
 Abrahamson's algorithm, for alphabet
 sizes smaller than $\sqrt{m/\log m}$, computes one convolution for every
 character in the alphabet. The convolutions proved to be expensive in practice, so
 Abrahamson's algorithm  was competitive only for large $k$.  
 Subset $k$-mismatches, applied for the full set
 of alignments, performed well for relatively small $k$. In most cases, the
 suffix array version was slower than the one with $O(m^2)$ time pre-processing, 
 because of the added $O(\log n)$ factor when searching in the suffix array. 
 Knapsack $k$-mismatches was the fastest among the algorithms compared because
 on most test instances the knapsack could be filled within the given budget.
 On such instances the algorithm did not perform convolution operations.

\begin{figure*}[t]
\scalebox{.55} {
\begin{tabular}{ccc}
\includegraphics{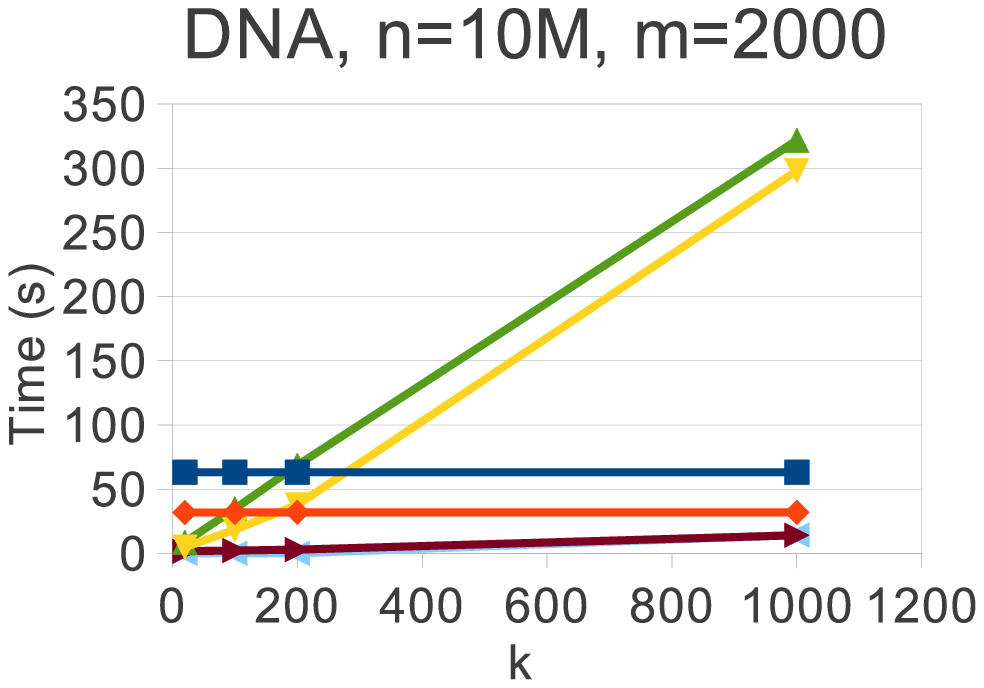} &
\includegraphics{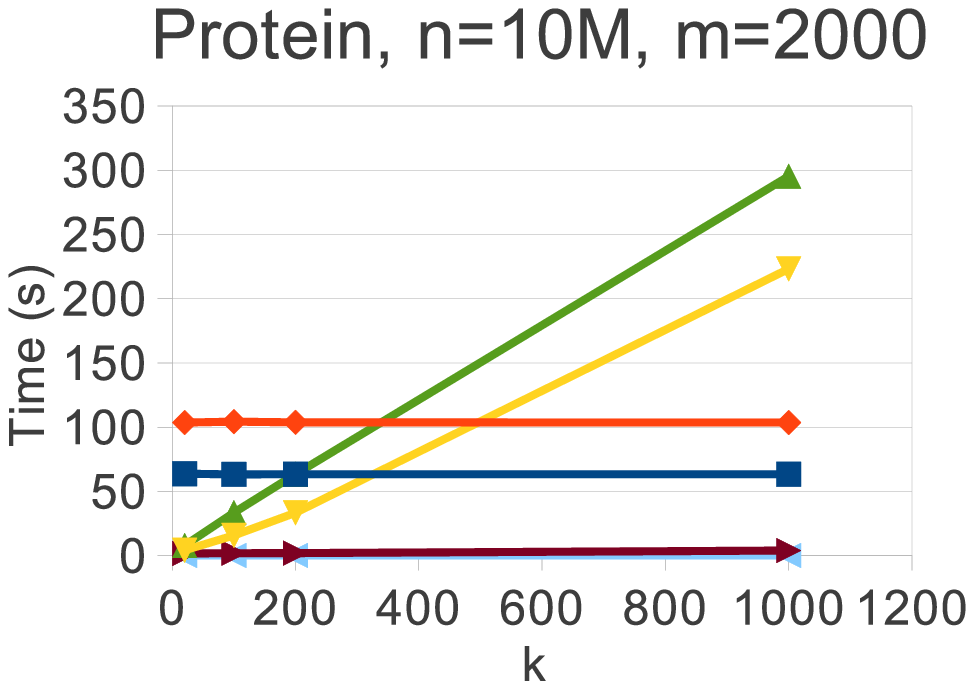} & 
\includegraphics{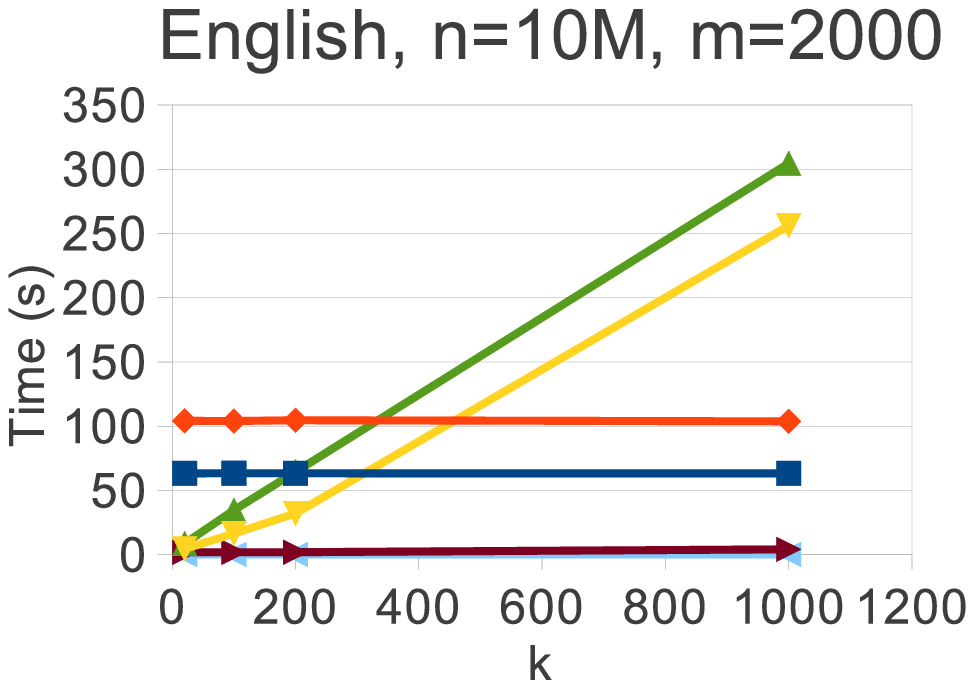} \\
\includegraphics{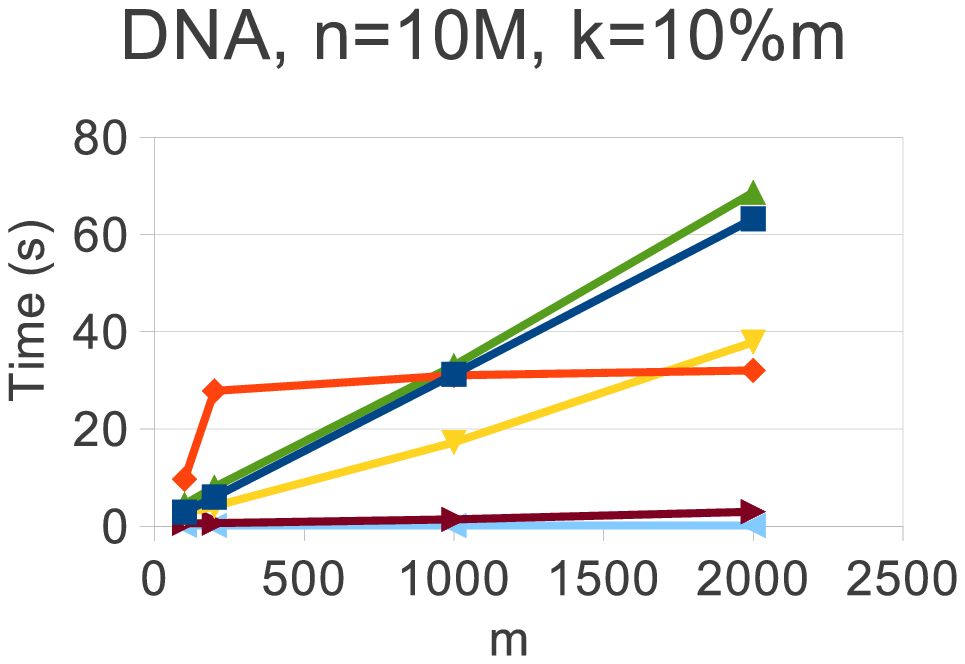} &
\includegraphics{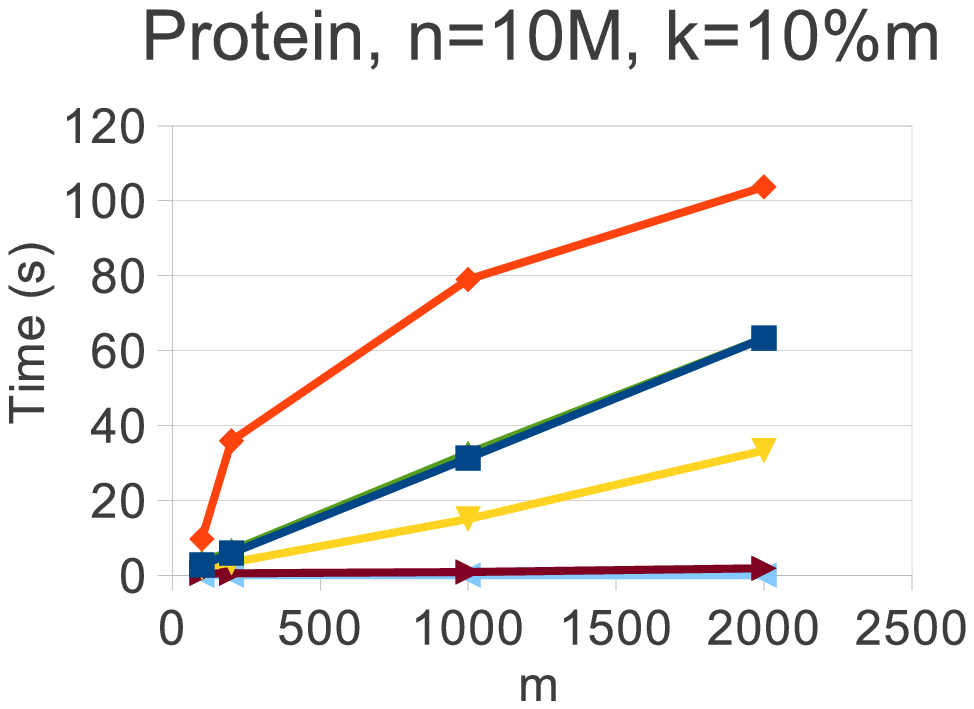} &
\includegraphics{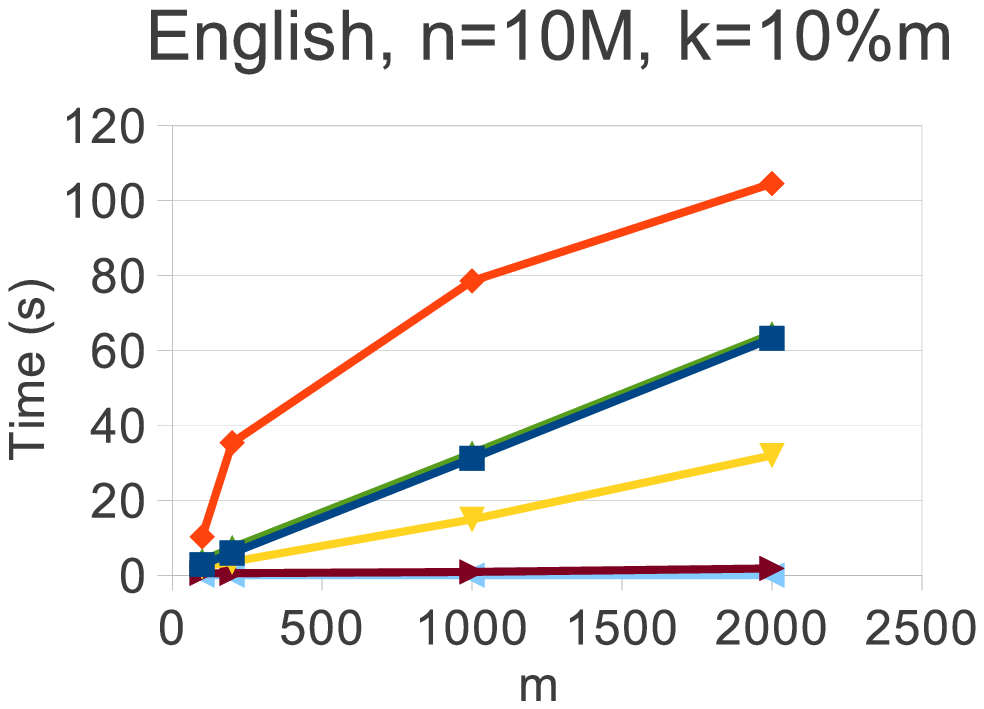} \\
\includegraphics{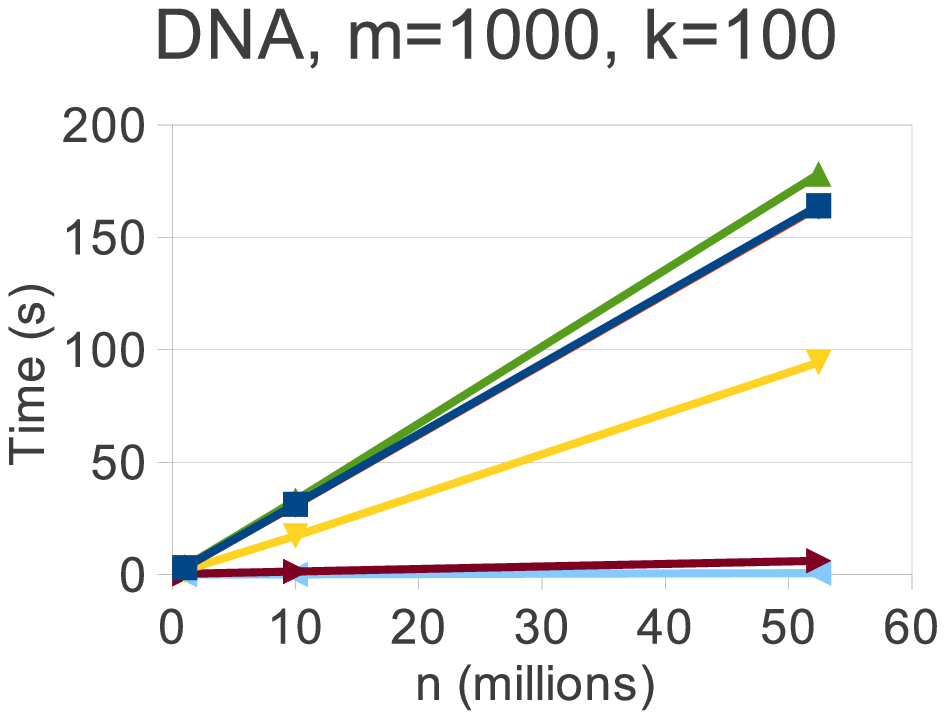} &
\includegraphics{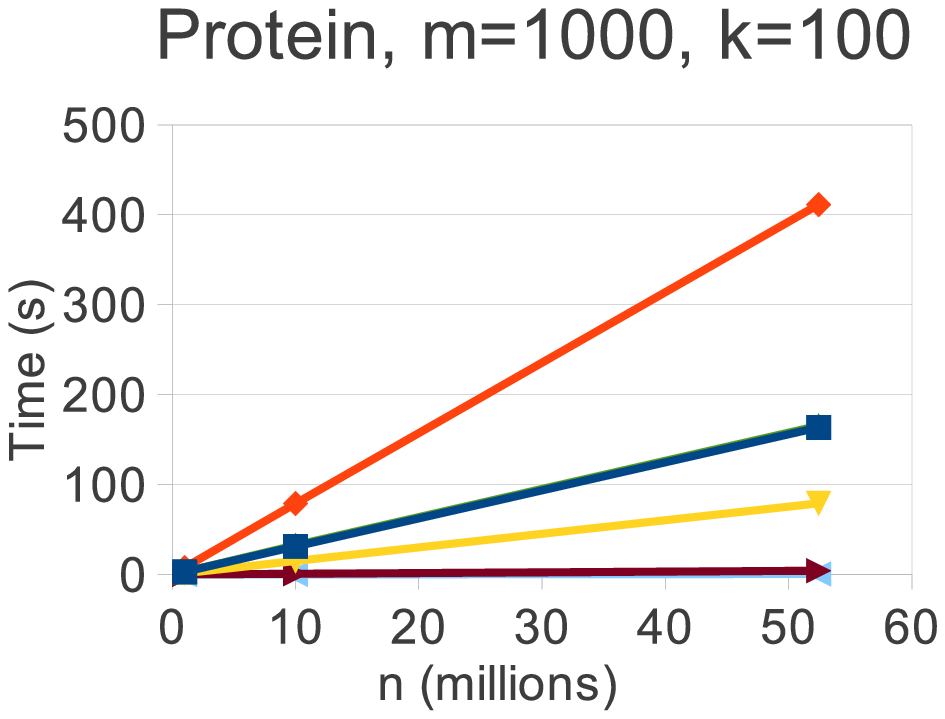} &
\includegraphics{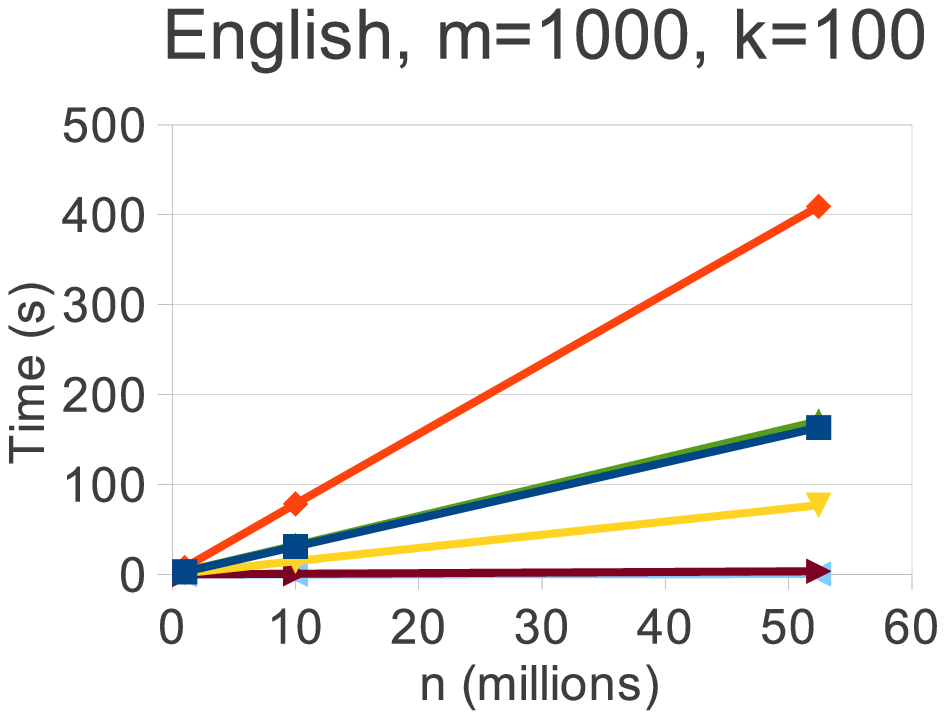} \\
\includegraphics{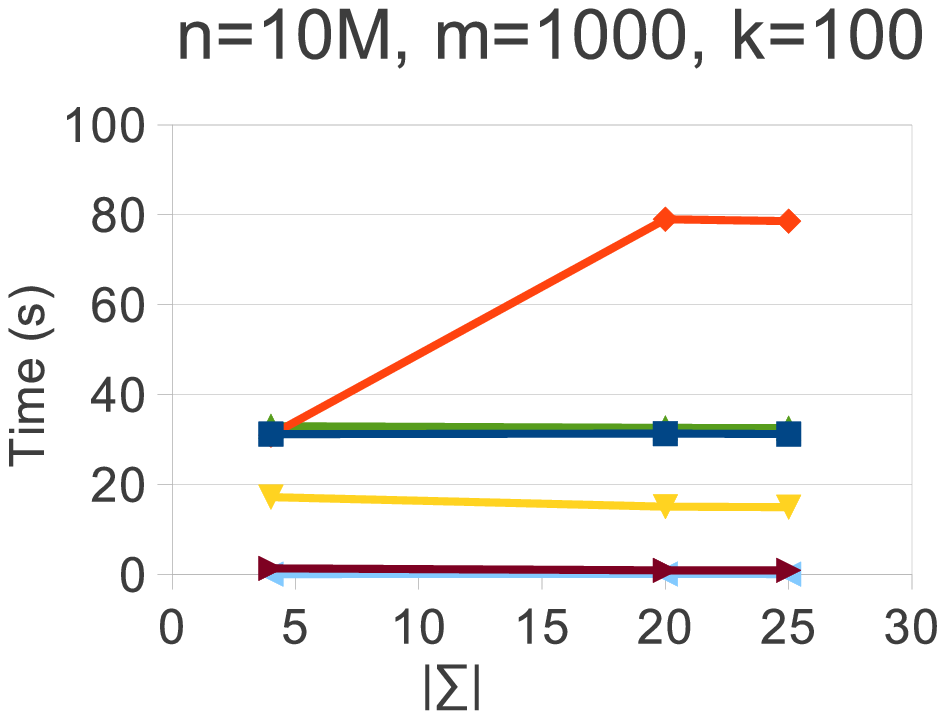} &
\includegraphics{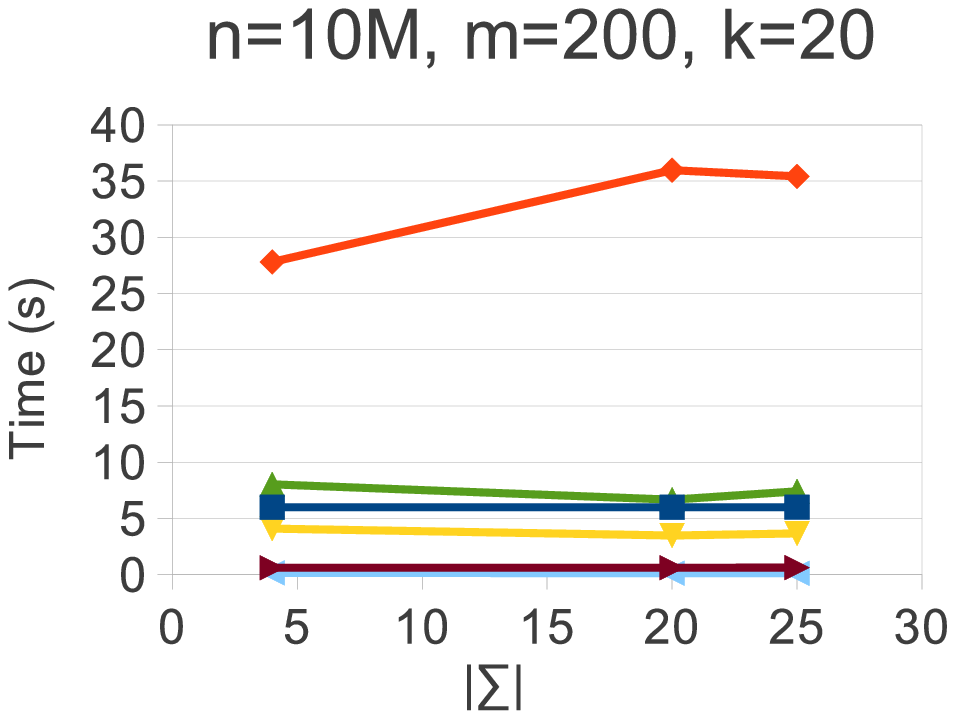} &
\includegraphics{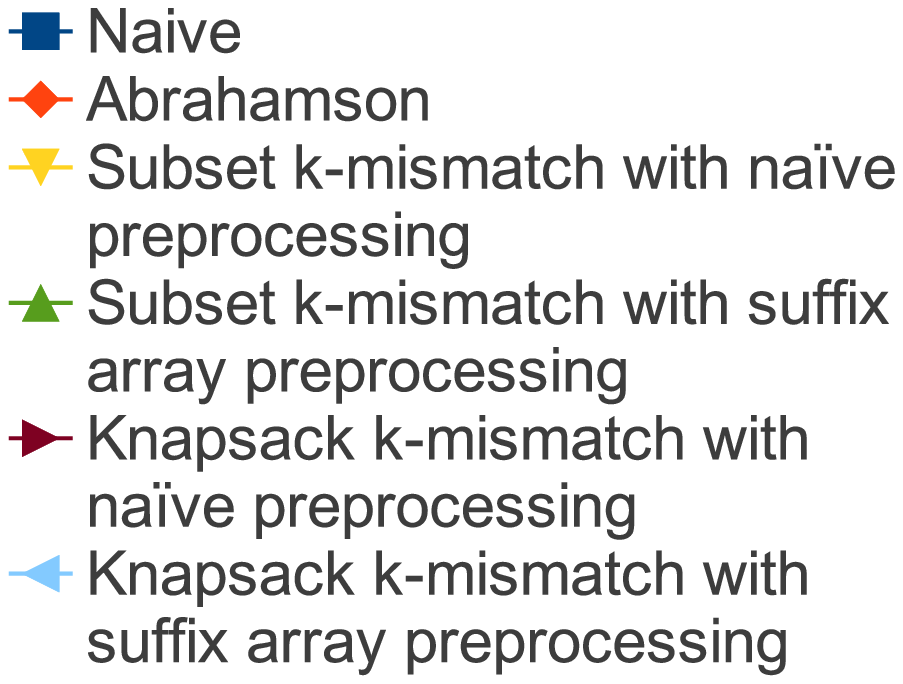} \\
\end{tabular}
}
\caption{Run times for the tested algorithms. M stands for
millions. Top row: k varies. Second row: m varies.
Third row: n varies. Bottom row: alphabet size varies, legend. }
\label{fig_runtimes} 
\end{figure*}

\section{Conclusions}
We have  introduced several randomized and deterministic, exact and approximate
algorithms for pattern matching with mismatches and the $k$-mismatches
problems, with or without wild cards. These algorithms improve the run time, 
simplify, or extend previous algorithms to accommodate wild cards. 
We implemented the deterministic algorithms. An empirical comparison of these 
algorithms showed that those based on
character comparison outperformed those based on convolutions.

\commentOut{
\section{Authors contributions}
SR designed the randomized and approximate algorithms. MN designed
and implemented the deterministic algorithms and carried
out the empirical experiments. MN and SR drafted and approved the manuscript.
}


%
\section*{Acknowledgment}
This work has been supported in part by the following grants: NSF 0829916 and NIH R01LM010101.

\bibliographystyle{IEEEtran}
\bibliography{kmis-references}

\appendices

\section{Algorithm 6 - Subset $k$-mismatches}
\begin{center}
\label{alg_subsetk}
\begin{quote}
\begin{tabbing}
aa \= aa \= aa \= aa \= aa \= \kill
{\bf Algorithm 6}  Subset $k$-mismatches(S)\\
\> // S = set of positions to test\\
\> {\bf let} $M_i := 0$ {\bf for all} $i = 1,n$\\
\> $i:=1$ \\
\> {\bf while} {$i \leq n$} {\bf do} \\
\> \> {\bf find} the largest $l$ such that $\exists j$ for which\\
\>\>\> $T_{i..i+l-1} = P_{j..j+l-1}$ \\
\> \> {\bf for all} {$s \in S$} {\bf where} {$s \leq i < s + m$} {\bf
do}
\\
\> \> \> $M_s = ${\bf updateMism}$(M_s, i-s+1, j, l)$ \\
\> \> \> {\bf if} {$M_s > k$} {\bf then} {$S = S - \{s\}$} \\
\> \> {\bf end for} \\
\> \> $i = i + l + 1$ \\
\> {\bf end while} \\
\> {\bf return} $M$ \\
\\
{\bf function} {\bf updateMism}$(c, s_1, s_2, l)$ \\
\> {\bf while} {$l > 0$ and $c \leq k$}  {\bf do} \\
\> \> {$d := {\bf lca}(s_1,s_2)$} \\
\> \> {\bf if} {$d \geq l$} {\bf then} {\bf return} {$c$}\\
\> \> $c := c + 1$\\
\> \> $d := d + 1$\\
\> \> $s_1 := s_1 + d$\\
\> \> $s_2 := s_2 + d$\\ 
\> \> $l := l - d$\\
\> {\bf end while} \\
\> {\bf return} $c$ \\
\end{tabbing}
\end{quote}
\end{center}

\vspace{0.1in}
\section{Algorithm 7 - Knapsack $k$-mismatches}
\begin{center}
\label{alg_knapsack}
\begin{quote}
\begin{tabbing}
aa \= aa \= aa \= aa \= aa \= \kill
{\bf Algorithm 7}  Knapsack $k$-mismatches \\
\> {\bf compute} $F_i$ and $f_i$ for every $i \in \Sigma$ \\
\> {\bf sort} $\Sigma$ with respect to $F_i$ \\
\> $s := 0$\\
\> $c := 0$\\
\> $i := 1$\\
\> $B := n \sqrt{k \log k}$ \\
\> \bf{while} {$s < 2k$ {\bf and} $c < B$} {\bf do} \\
\> \> $t := \min(f_i, 2k - s)$ \\
\> \> $s := s + t$\\
\> \> $c := c + t * F_i$ \\
\> \> $i := i + 1$ \\
\> {\bf end while} \\
\> {$\Gamma := \Sigma[1..i]$} \\
\> {$M := {\bf Mark}(T, n, \Gamma)$} \\
\> {\bf if} {$s = 2k$} {\bf then} \\
\> \> $S := \{i | M_i \geq k\}$  \\
\> \> {\bf return Subset $k$-mismatches}($S$)\\
\> {\bf else} \\
\> \> {\bf for}	{$\alpha \in \Sigma - \Gamma$} {\bf do} \\
\> \> \> $C := {\bf convolution}(T^{\alpha}, P^{\alpha})$ \\
\> \> \> {\bf for} {$i:=1$ {\bf to} $n$} {\bf do} \\
\> \> \> \> {$M_i = M_i + C_i$} \\
\> \> {\bf end for} \\
\> \> $S = \{i | M_i \geq m-k\}$\\	  
\> \> {\bf return} S\\
\> {\bf end if} \\
\end{tabbing}
\end{quote}
\end{center}

\section{Proof of Theorem \ref{knapsack_less}}
Theorem \ref{knapsack_less} states that Knapsack $k$-mismatches
will spend at most as much time as the algorithm in \cite{ALP04} to do
convolutions and marking.

\begin{proof}
{\bf Observation:} In all the cases presented bellow,
Knapsack $k$-mismatches can have a run time as low as $O(n)$, for example if
there exists one character $\alpha$ with $f_\alpha = O(k)$ and $F_\alpha = O(n/k)$.

{\bf Case 1:} $|\Sigma| \geq 2k$. The algorithm in \cite{ALP04} chooses 2k
instances of distinct characters to perform marking. This ensures that the cost
$M$ of the marking phase is less or equal to $n$ and the number of
remaining positions after filtering is no more than $M/k$. Our algorithm puts in the
knapsack 2k instances of not necessarily different characters such that the cost
$B$ of the marking phase is minimized. Clearly $B \leq M$ and the number of
remaining positions after filtering is less or equal to $B/k \leq M/k$.

{\bf Case 2:} $|\Sigma| < 2\sqrt{k}$. The algorithm in \cite{ALP04} performs
one convolution per character to count the total number of matches for every
alignment, for a run time of $O(|\Sigma|n\log m)$. Only in the worst case, 
Knapsack $k$-mismatches cannot fill the knapsack at a cost $B < |\Sigma|n\log
m$ so it defaults to the same run time. However, in the best case, the
knapsack can be filled at a cost $B$ as low as $O(n)$ and so the run time could be linear.

{\bf Case 3:} $2\sqrt{k} \leq |\Sigma| \leq 2k$. A symbol that appears in the
pattern at least $2\sqrt{k}$ times is called frequent.

{\bf Case 3.1:} There are at least $\sqrt{k}$ frequent symbols. The algorithm in
\cite{ALP04} chooses $2\sqrt{k}$ instances of $\sqrt{k}$ frequent symbols to do
marking and filtering at a cost $M \leq 2n\sqrt{k}$. Since Knapsack
$k$-mismatches will minimize the time $B$ of the marking phase we have $B
\leq M$ so in the worst case the run time is the same as for \cite{ALP04}.

{\bf Case 3.2:} There are $A < \sqrt{k}$ frequent symbols. The
algorithm in \cite{ALP04} first performs one convolution for each frequent
character for a run time of $O(A n \log m)$.
Two cases remain:

{\bf Case 3.2.1:} All the instances of the non-frequent symbols number less than
$2k$ positions. The algorithm in \cite{ALP04} replaces all instances of frequent
characters with wild cards and applies a $O(n\sqrt{g} \log{m})$ algorithm to
count mismatches, where $g$ is the number of non-wild card positions. Since
$g<2k$ the run time for this stage is $O(n\sqrt{k} \log{m})$ and the total
run time is $O(An\log m + n\sqrt{k} \log{m})$.
Knapsack $k$-mismatches can always include in the knapsack all the
instances of non frequent symbols since their total cost is no more than
$O(n\sqrt{k})$ and in the worst case do convolutions for the remaining
characters at a total run time of $O(An\log m + n\sqrt{k})$. In practice,
the knapsack will be filled using some instances of both frequent and
infrequent characters, whichever minimize the cost.

{\bf Case 3.2.2:} All the instances of the non-frequent symbols number at
least $2k$ positions. The algorithm in \cite{ALP04} chooses $2k$ instances of
infrequent characters to do marking. Since each character has frequency less
than $2\sqrt{k}$, the time for marking is $M < 2n\sqrt{k}$ and there are no more
than $M/k$ positions left after filtering. Knapsack
$k$-mismatches chooses characters in order to minimize the time $B$
for marking, so $B \leq M$ and there are no more than $B/k \leq M/k$ positions
left after filtering.
\end{proof}

\ifCLASSOPTIONcaptionsoff
  \newpage
\fi

\end{document}